\documentclass[runningheads,a4paper]{llncs}

\usepackage{amssymb}
\usepackage{amsmath}
\usepackage{graphicx}
\usepackage[utf8]{inputenc}

\newcommand\Oh{\ensuremath{\mathcal{O}}}
\newcommand{\pred}[1]{\mathrm{pred}{(#1)}}

\begin{document}

\mainmatter  % start of an individual contribution

\title{Chain minors are FPT}

\titlerunning{Fixed Parameter Tractable algorithm for Chain Minor problem}

\author{Jarosław Błasiok$^1$%
\and Marcin Kamiński$^2$}
\authorrunning{J. Błasiok /and M. Kamiński}

\institute{
$^1$Instytut Informatyki\\Uniwersytet Warszawski\\ {\tt jb291202@students.mimuw.edu.pl}\\\vspace{.3cm}
$^2$Département d'Informatique\\Universit\'e libre de Bruxelles\\and\\Instytut Informatyki\\Uniwersytet Warszawski\\{\tt mjk@mimuw.edu.pl}
}
% "llncs.dem" accompanies the document class "llncs.cls".parameterized
%

\toctitle{Lecture Notes in Computer Science}
\tocauthor{Authors' Instructions}
\maketitle

\begin{abstract}
Given two finite posets $P$ and $Q$, $P$ is a \emph{chain minor} of $Q$ if there exists a partial function $f$ from the elements of $Q$ to the elements of $P$ such that for every chain in $P$ there is a chain $C_Q$ in $Q$ with the property that $f$ restricted to $C_Q$ is an isomorphism of chains. \smallskip

We give an algorithm to decide whether a poset $P$ is a chain minor of a poset $Q$ that runs in time $\Oh(|Q| \log |Q|)$ for every fixed poset $P$. This solves an open problem from the monograph by Downey and Fellows [\emph{Parameterized Complexity}, 1999] who asked whether the problem was fixed parameter tractable.

%	We study certain $\preceq$ relation on finite posets called {\it chain minor}, and in particular \textsc{Chain Minor Ordering} problem of deciding for given posets $P$ and $Q$ $P\preceq Q$. We prove that \textsc{Chain Minor Ordering} problem admits FPT solution, when parameterized by $|P|$; our algorithm has time complexity $\Oh(f(|P|)|Q| \log |Q|)$, and linear space complexity, therefore solving open question by Downey and Fellows \cite{Downey}. We also present randomized version of this algorithm having slightly better time complexity (with respect to $|Q|$), namely $\Oh(f(|P|)|Q|)$.
	
%	It has long been known, due to work of Gustedt \cite{Gustedt}, that {\it chain minor} relation is well-quasi-order; such relations became recently widely studied, mainly because every set of objects closed under wqo relation could be characterized by finite set of forbidden objects with respect to that relation. In particular, presented algorithm for \textsc{Chain Minor Ordering} yields linear algorithm for recognition problem for any family closed under chain minor.
\end{abstract}\bigskip

\section{Introduction}
It is widely believed that NP-hard problems do not admit polynomial-time deterministic algorithms. Nevertheless, such problems tend to appear in practical applications and it is necessary to deal with them anyway. Among many approaches to NP-hard problems {\it parameterized complexity} has recently received a lot of attention. It was first studied systematically by Downey and Fellows in \cite{DowneyFellows}. The main idea of parameterized complexity is to equip the instance of a problem with a parameter and confine the superpolynomial behaviour of the algorithm to the parameter. Here we can efficiently solve large instances of the problem as long as the parameter is small.\medskip

\noindent{\it Parameterized complexity}. More formally, an instance of a parameterized problem is a pair $(I, k)$ where $k\in \mathbb{N}$. XP is the class of parameterized problems such that for every $k$ there is an algorithm that solves that problem in time $\Oh(|I|^{f(k)})$, for some function $f$ (that does not depend on $I$). One example is the \textsc{Clique} problem parameterized by the size of the clique defined as follows: given $(G, k)$ where $G$ is graph and $k$ is a natural number, is there a clique of size $k$ in $G$? One can simply enumerate all $k$-subsets of vertices to solve the problem in time $\Oh(n^{k+2})$, hence, in time polynomial for every fixed $k$.

Much more desirable parameterized complexity is FPT. A parameterized problem is called {\it fixed parameter tractable} (FPT) if there is an algorithm that for every instance $(I, k)$ solves the problem in time $\Oh(f(k) n^c)$ for some function $f$ (that does not depend on $n$). That is, for a fixed parameter $k$, then problem is solvable in polynomial time and the degree of the polynomial does not depend on $k$. \textsc{Satisfability} of boolean formula parameterized by number of variables is FPT; it can be solved by a brute force algorithm in time $\Oh(2^k m)$ where m is size of instance.

Downey and Fellows in their monograph \cite{DowneyFellows} included a list of open problems, asking whether they admit an FPT solution (``FPT suspects'') or are hard by means of parameterized complexity (``tough customers''). Recently, Fomin and Marx have revised this list of problems \cite{FominMarx}. Many of the problems from the original list have been solved since the publication of \cite{DowneyFellows}, yet \textsc{Chain minor} remains open. It was listed as a ``tough customer'' -- suspecting it is not fixed parameter tractable. However, we  prove otherwise. \smallskip

\noindent{\it Chain minors}. Chain minors were introduced by M\"oring and M\"uller in \cite{MohringMuller} in the context of scheduling stochastic project networks and first studied systematically by Gustedt in \cite{GustedtPhD} and in his PhD thesis \cite{Gustedt}. Gustedt proved that finite posets are \emph{well quasi ordered} by chain minors, that is, in any infinite sequence of posets there is a pair of posets such that one is a chain minor of the other. A consequence of this fact is that any class of graphs closed under taking chain minors can be characterized by a finite family of minimal forbidden posets. 

The {\sc Chain minor} problem is to decide, given two posets $P$ and $Q$, whether $P$ is a chain minor of $Q$. The parameterized approach to {\sc Chain Minor} is justified as Gustedt showed in \cite{GustedtPhD} that {\sc Chain Minor} is NP-hard (giving a reduction from \textsc{Precendence Constrained Scheduling}). Note that it is not known whether \textsc{Chain Minor} is NP-complete. There is no obvious nondeterministic polynomial-time algorithm for that problem, except for a very simple case --- Gustedt in his PhD thesis has proved that \textsc{Chain Minor}  is NP-complete when restricted to posets of height at most 3. \medskip

\noindent{\it Our results}. Gustedt also gave an XP algorithm for the \textsc{Chain Minor} problem \cite{Gustedt}. More specifically, he gave an algorithm that checks whether $P$ is a chain minor of $Q$ in time $\Oh(|P|^2 |Q|^{|P|} + f(|P|))$. We improve his result, giving two fixed parameter tractable algorithms (parameterized by $|P|$) --- randomized and deterministic --- where the former one runs in $\Oh(f(|P|) |Q|)$ time and the latter in $\Oh(f(|P|) |Q| \log |Q|)$ time. Both algorithms need linear memory.

The technique that we use to design the FPT algorithm is called {\it color coding} and was originally developed by Alon, Yster, and Zwick in \cite{colorCoding} to give the first FPT algorithm for the \textsc{k-Path} problem (= finding a path of size $k$ in a given graph). Since then, this technique has been successfully applied many times, yet in most of those examples colors where introduced artificially (as in \textsc{k-Path}). In our case, they are naturally derived from the problem definition.

\section{Definitions and basic facts}
A finite partially ordered set (\emph{poset}) is a pair $(V, <)$ where $V$ is a finite set and $<$ is a binary relation on $V$ that is transitive, irreflexive, and antisymmetric. A \emph{chain} in a poset is a sequence of elements $(v_1, v_2, \ldots v_n)$, $v_i \in V$ such that $v_i < v_j$, for all $1 \leq i < j \leq n$.

Given two finite posets $P=(V_P, <_P)$ and $Q=(V_Q, <_Q)$, we say that $P$ is a {\it chain minor} of $Q$ ($P \preceq Q$) if and only if there exists a partial function $f: V_Q \longrightarrow V_P$ with a property that for every chain $(c_1, c_2, \ldots, c_n)$ in $P$ there is a chain $(c'_1, c'_2, \ldots c'_n)$ in $Q$ such that $f(c'_i) = c_i$. In this case, $f$ a {\it witness} for $P \preceq Q$ and we write $P \preceq^f Q$. It is easy to check that $\preceq$ is a quasi-order (transitive and antisymmetric). One can easily check that if $V_P \subseteq V_Q$ and $<_P$ is induced by $<_Q$ (that is, $P$ is subposet of $Q$), then $P$ is also a chain minor of $Q$.

\begin{figure}[h]
\begin{center}
\includegraphics[scale=0.2]{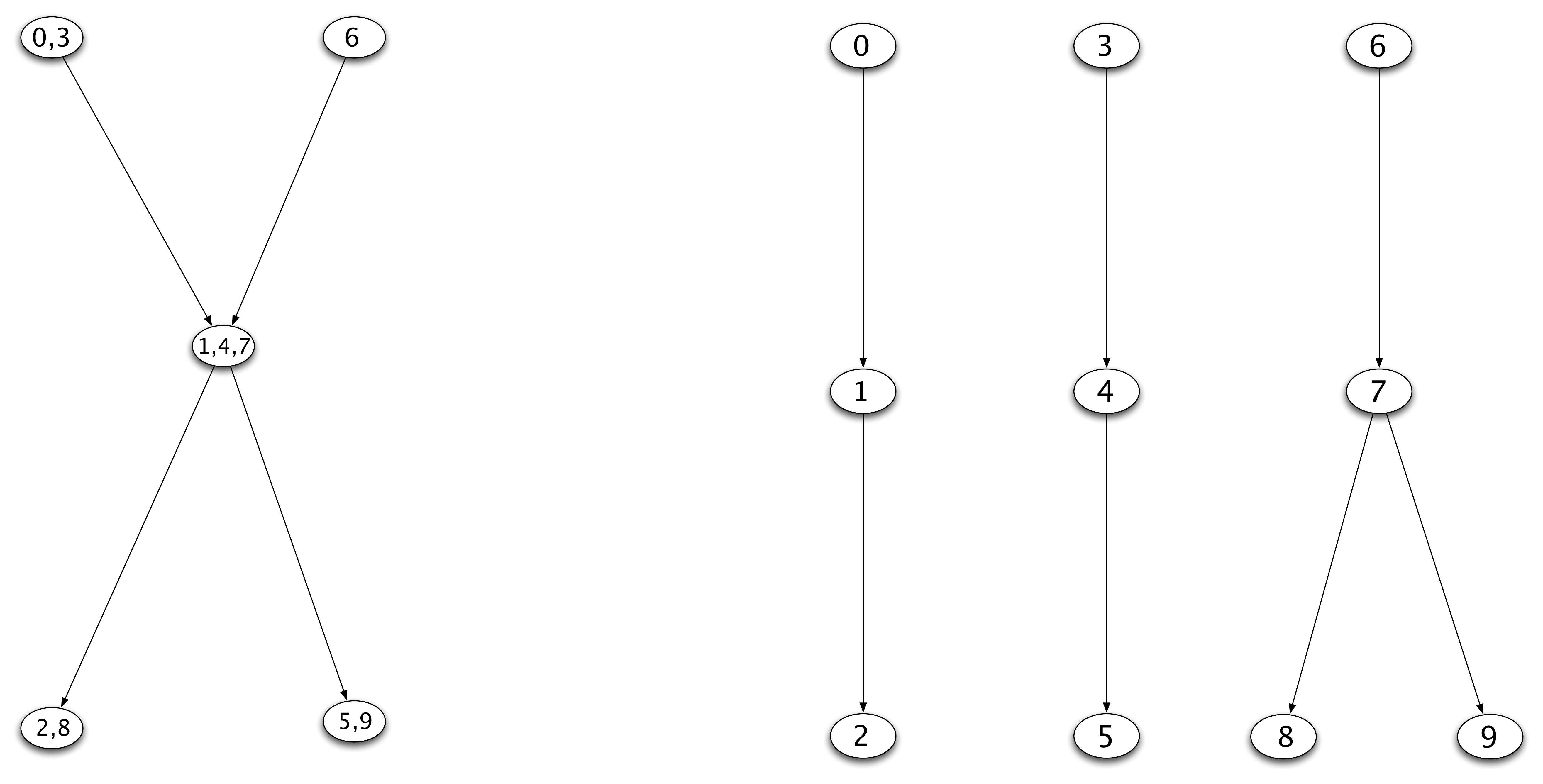}
\end{center}
\caption{The left poset $P$ is a chain minor of  the right poset $Q$ as certified by the witness function from the elements of the $Q$ to the elements of $P$.}
\label{fig-chainminor}
\end{figure}

%We will present a useful lemma, originally proved by Gustedt \cite{Gustedt} in a slightly different form as Lemma 4.1.
%
%LEMMA???

\section{Algorithm}

Our goal is to present a deterministic FPT algorithm. We will start with a randomized algorithm and  use a standard technique (of splitters) to derandomize it at the price of slightly worse time complexity. However, we need some auxiliary lemmas first.

\begin{lemma}
	\label{detlem}
	There is a deterministic algorithm which given two posets $P$ and $Q$ and a partial function $f: V_Q \longrightarrow V_P$ determines whether $P \preceq^f Q$ in time $\Oh(2^{|P|} |Q|)$.
\end{lemma}
\begin{proof}
	For $q \in Q$, let $\pred{q}$ be the set of elements less or equal to $q$ in $Q$, that is, $\pred{q} = \{q' \in Q : q' \leq Q\}$.
	It is enough to iterate over  all chains of $P$, and for every chain $(c_1, c_2, \ldots, c_p)$  consider only those vertices of $V_Q$ which are mapped by $f$ to any of $c_i$ --- let us call them $Q'$. Let us now consider vertices from $Q'$ in topological order. For every vertex $q$, let us compute the maximum $j$ such that one can find a chain $c'_1, c'_2, \ldots, c'_j$ in set $\pred{q}$ (as usual, we demand $f(c'_i) = c_i$). Let us call that value $\mathrm{maxc}(q)$.

	To calculate $\mathrm{maxc}(q)$ knowing $\mathrm{maxc}$ of every predecessor, we just take
	\begin{equation*}
		\mathrm{maxc}(q) = \left\{ \begin{array}{ll}
				j & \;\;\;\; \mathrm{if} \max_{v<q} \mathrm{maxc}(v) = j-1 \land f(q)=c_j \\
				\max_{v<q} \mathrm{maxc}(v) & \;\;\;\; \mathrm{otherwise}
			\end{array}\right.
	\end{equation*}
The solution can be read off from $\mathrm{maxc}$ values.\qed
\end{proof}

\begin{lemma}
	\label{SmallUniverse}
	Let $P$ and $Q$ be finite posets and $k=|P|$. If $P \preceq^f Q$, then there is a subposet $Q^0$ of $Q$ of size at most $2^k k$ such that if $f'$ is equal to $f$ on $Q^0$, then $P \preceq^{f'} Q$.
\end{lemma}
\begin{proof}
	Let $P \preceq^f Q$ and let $\cal C$ be the set of all chains in $\cal C$. $\cal C$ has at most $2^k$ elements (as any subset of the elements from $V_P$ forms at most one chain). For every chain $c = (c_1, \ldots, c_{n_c})\in {\cal C}$, take an arbitrary chain $(c'_1, \ldots, c'_{n_c})$ in $Q$ such that $f(c'_i) = c_i$, for $i=1, \ldots, n_{n_c}$.
	
	 Now let $V_{Q^0}$ be $\bigcup_c \{ c'_1, c'_2, \ldots, c'_{n_c}\}$. Notice that $|V_{Q^0}| \leq \sum_{c \in {\cal C}} n_c \leq \sum_{c \in {\cal C}} k \leq 2^k k$. If  $f'$ is equal to $f$ on $Q^0$, we have to check that given a chain $c_1, c_2, \ldots c_{n_c}$ in $P$ one can find preimages with respect to $f'$ of the elements of that chain such that the preimages form a chain in $Q$. It suffices to take the elements $c'_i$ from above; they belong to $Q^0$ by definition, thus $f'(c'_i) = f(c'_i) = c_i$ for  $i=1, \ldots, n_c$ and the elements $c'_1, \ldots, c'_n$ were chosen to be a chain.\qed
\end{proof}

\subsection{Randomized algorithm}
Now we will state and prove a key lemma for Theorem \ref{thmRandom}.
\begin{lemma}
	\label{prob}
	If $P \preceq Q$, then a function $g: V_P  \longrightarrow V_Q$ taken uniformly at random from the set of all such functions is a witness for $P \preceq Q$ with probability at least $k^{-2^k k}$, where $k = |P|$.
\end{lemma}
\begin{proof}
	Let $f$ be a witness for $P \preceq Q$. Now take $Q^0$ as in Lemma \ref{SmallUniverse}. It follows from Lemma  \ref{SmallUniverse} that it is sufficient to show that a function $g$ taken uniformly at random is equal to $f$ on $Q^0$ with high probability, as the probability of $g$ being a witness for $P \preceq Q$ is at least as large. Now the lemma follows from the following simple calculation.

	\begin{eqnarray*}
		\mathbb{P}_g(P \preceq^g Q) & \geq & \mathbb{P}_g(g|Q^0 = f|Q^0) \\
									& = & \prod_{v\in Q^0} \mathbb{P}(g(v) =  f(v)) \\
									& = & \prod_{v\in Q^0} \frac{1}{k} \\
								    & = & k^{-2^k k}
	\end{eqnarray*}\qed
\end{proof}

Now we are ready to prove Theorem \ref{thmRandom}.

\begin{theorem}
	\label{thmRandom}
	There is a randomized algorithm for \textsc{Chain Minor} with time complexity $\Oh({|P|}^{2^{|P|} {|P|}} |Q|)$ and linear space complexity.
\end{theorem}

\begin{proof}
Let $k = |P|$. It is enough to repeat the following procedure $k^{2^k k}$ times: take a random function $g$ and check whether it is a witness for $P \preceq Q$. If any of those function is a witness, return \textsc{Yes}; otherwise, return \textsc{No}. The desired time and space complexity follow from Lemma \ref{detlem}. Lemma \ref{prob} bounds the probability of an error by a constant. Indeed, if $P\preceq Q$, then the probability that the algorithm answers \textsc{No} is not greater then $(1 - 1/p_k)^{p_k}$, where $p_k = k^{2^k k}$, which is bounded by $(1 - 1/2)^2$, for $k\geq 2$ (and tends to $1/e$ as $k$ tends to infinity).\qed
\end{proof}

\subsection{Deterministic algorithm}
We will derandomize the algorithm from Theorem \ref{thmRandom} using a well-known derandomization technique of splitters. A $(n,k,l)$-splitter is a family of functions $\mathcal{F}$, $\mathcal{F} \ni f : \{1, \ldots n \}  \longrightarrow \{1, \ldots, k\}$, such that for every $W\subseteq \{1\ldots n\}$ there is some function $f\in \mathcal{F}$ which is injective on $W$. We will need the following theorem by Naor, Schulman, and Srinivasan from \cite{Naor}.

\begin{theorem}{(\cite{Naor})} There exists a $(n, k, k)$-splitter that can be constructed in time $\Oh(e^k k^{\Oh(\log k)} n \log n)$.
\end{theorem}

\begin{theorem}
	\label{thmDetermin}
	There is a deterministic algorithm for \textsc{Chain Minor} with time complexity $\Oh(f(|P|) |Q| \log |Q|)$ and linear space complexity.
\end{theorem}

\begin{proof}
Given $P$ and $Q$, let us take $k=|P|$, $n=|Q|$. Fix a bijection between $\{1\ldots n\}$ and $V_Q$. Now we can just iterate through every $(n, 2^k k, 2^k k)$ splitter and every function from the set $\{1, \ldots, 2^k k\}$ to $P$, and check whether the composition of these two functions is a witness for $P \preceq Q$. 

To prove correctness of the algorithm, let us consider $P \preceq^f Q$ and take $Q^0$ as in Lemma $\ref{SmallUniverse}$. It follows from the definition of splitters that there exists a function $f$, such that $f$ is injective on $Q^0$. Then, just because we iterate over all functions from the set $\{ 1, \ldots k 2^k\}$ to $P$ at some point we take one, such that the composition equals  $f$ when restricted to $Q^0$. This pair yields a witness for $P \preceq Q$.\qed
\end{proof}

\section{Conclusions}

\begin{enumerate}
\item It is easy to prove that every class of posets closed under taking of chain minors can be characterized by a set of minimal forbidden chain minors. Gustedt proved in \cite{Gustedt} that posets are well quasi ordered. Consequently, each such set of forbidden chain minors is finite. Gustedt also gave an XP algorithm to decide whether a poset $H$ is a chain minor of a poset $Q$ when parameterized by the number of elements of  $H$. These two results show that for every class of posets $\cal P$ closed under taking chain minors there exists a polynomial-time algorithm deciding whether the input poset $Q$ is in $\cal P$. (The exponent of the polynomial depends on the class.) 

\hspace*{.5cm} We give an FPT algorithm to test whether a poset $H$ is a chain minor of a poset $Q$ when parameterized by the number of elements of  $H$. A consequence of our result is that for every class of posets $\cal P$ closed under taking chain minors there exists a $\Oh(|Q| \log |Q|)$ algorithm deciding whether the input poset $Q$ is in $\cal P$. \medskip

\item The project of Graph Minors of Robertson and Seymour is arguably one of the most significant achievements in modern graph theory. Robertson and Seymour proved that graphs are well quasi ordered under graph minors and gave an FPT algorithm to decide whether a graph $H$ is a minor of a graph $G$ when parameterized by $H$. They were also able to describe the structure of graphs that do not contain a fixed graph as a minor.  

\hspace*{.5cm} Is there a parallel theory possible for chain minors in posets? Gustedt proved in \cite{Gustedt} that chain minors are well quasi ordered and this work gives an FPT algorithm for the {\sc Chain Minor} problem. However, neither of the two elucidates the structure of posets with a forbidden chain minor. Is a structural characterization possible?

\hspace*{.5cm} In particular, it looks that characterizing posets without $p C_q$ as a chain minor is already the first challenge. ($p C_q$ is a poset consisting of $p$ disjoint chains each on $q$ vertices.) Note that any poset of size $p$ and height $q$ is a chain minor of $2^p C_q$. It is also quite straightforward that posets without $C_q$ chain minor are just posets of height less then $q$ but even a characterization of posets without $2C_q$ as a chain minor seems elusive. \medskip

\item Let us recall that Gustedt showed in \cite{Gustedt} that the \textsc{Chain Minor} problem is NP-hard but it is not known whether the problem is NP-complete. This is an interesting question. In particular, given two posets $P$, $Q$ and a function $w: Q  \longrightarrow P$, is there a polynomial-time deterministic algorithm deciding whether $w$ is a witness for $P\preceq Q$? Such algorithm would naturally give rise to an NP algorithm for \textsc{Chain Minor}.\medskip

\item At last, both our algorithms are double exponential in the parameter. Could this be improved to get a single exponential dependence?
\end{enumerate}



\begin{thebibliography}{}

\bibitem{colorCoding}
Alon N., Yuster R., Zwick U.:
\newblock \emph{Color-coding}.
\newblock Journal of the ACM, 42(4), p. 844 -- 856, 1995

\bibitem{DowneyFellows}
Downey, R.G., Fellows, M.R.:
\newblock Parameterized Complexity. 
\newblock Springer-Verlag, NewYork (1999)

\bibitem{FominMarx} 
Fomin, F.V., Marx D.:
\newblock \emph{FPT suspects and tough customers: Open problems of Downey and Fellows}.
\newblock Submitted.

\bibitem{GustedtPhD} 
Gustedt J.:
\newblock \emph{Algorithmic Aspects of Ordered Structures}.
\newblock PhD thesis, Berlin, 1992

\bibitem{Gustedt} 
Gustedt J.:
\newblock \emph{Well Quasi Ordering Finite Posets and Formal Languages}.
\newblock Journal of Combinatorial Theory, Series B, 65(1), p. 111 -- 124, 1995

\bibitem{MohringMuller} 
M\"ohring, R.H, M\"uller R.:
\newblock \emph{A combinatorial approach to obtain bounds for stochastic project networks}.
\newblock Tech. report, Technische Universit\"at Berlin, 1992

\bibitem{Naor} 
Naor M., Schulman L. J., Srinivasan A.:
\newblock \emph{Splitters and near-optimal derandomization}.
\newblock Proceedings of the 36th Annual Symposium on Foundations of Computer Science FOCS, 1995

\end{thebibliography}
\end{document}